\documentclass[12pt]{amsart}
\usepackage{amssymb}
\usepackage{hyperref}
\usepackage{amsmath}
\usepackage{amscd}
\usepackage{graphicx}
\usepackage{xcolor}
\usepackage{bm}
\usepackage{comment}

\DeclareMathSymbol{\shortminus}{\mathbin}{AMSa}{"39}

\newcommand{\group}[1]{\mathrm{#1}}

\newcommand{\rep}[1]{\rep{#1}}

\newcommand{\Tr}{\operatorname{Tr}}

\newcommand{\End}{\operatorname{End}}

\newcommand{\C}{\mathbb{C}}
\newcommand{\D}{\mathbb{D}}

\newcommand{\E}{\mathbb{E}}
\newcommand{\N}{\mathbb{N}}
\renewcommand{\P}{\mathbb{P}}
\renewcommand{\H}{\mathbb{H}}

\newcommand{\R}{\mathbb{R}}
\newcommand{\U}{\mathbb{U}}

\newcommand{\W}{\mathbf{W}}

\newcommand{\X}{\mathfrak{X}}
\newcommand{\Y}{\mathfrak{Y}}

\newtheorem{thm}{Theorem}[section]
\newtheorem{prop}[thm]{Proposition}

\newtheorem{cor}[thm]{Corollary}

\theoremstyle{definition}

\title[Disordered Schur Measures]{Disordered Schur Measures}

\author{Jonathan Novak}

\begin{document}


\maketitle

\section*{Introduction}
Okounkov's Schur measures are multiparameter generalizations of the geometric distribution from 
integers to integer partitions \cite{Okounkov}. They provide model examples of biorthogonal ensembles, 
a basic class of determinantal point processes isolated by Borodin \cite{Borodin}.
There are many ways in which one may specialize the parameters of Schur measures 
to obtain particular distributions of interest \cite{BBW,OkounkovUses,Walsh}. Here, 
we instead randomize parameters to obtain disordered Schur measures
which exhibit behavior reminiscent of spin glasses. 

We build quenched disorder into a Schur measure
by sampling its parameters from the Circular Unitary Ensemble \cite{Forrester},
and proceed to study thermodynamic features of the resulting disordered system. 
We first compute the quenched and annealed free energies of 
CUE-Schur measure and show that they are strictly 
separated in the thermodynamic limit, with disorder gap given by an explicit analytic function of fugacity.
We then give a representation-theoretic formula for moments of the partition function, and
show that the free energy converges in distribution to a random analytic function of fugacity
whose coefficients are iid exponential random variables. 

Next, we pass to a double scaling regime in which fugacity approaches its critical value as the 
particle number increases. We show that in this near-critical scaling the free energy becomes extensive
and the free energy density self-averages, with thermodynamic limit given by 
an explicit non-random function of the critical scaling parameter. Finally, we combine
self-averaging with the fact that the free energy is a completely 
monotone function of the critical parameter to obtain a quenched law of large numbers for the 
size of a random partition sampled from near-critical CUE-Schur measure. 

These results paint a thermodynamic picture of CUE-Schur measures which 
supports the spin glass analogy. Going further, a natural next step would be to study
the size of the overlap of two conditionally independent replicas from CUE-Schur measure.
In another direction, it would be interesting to analyze the distribution of the largest
part of a single sample --- for nonrandom
Schur measures this statistic has Tracy-Widom fluctuations in some specializations \cite{Matsumoto}
but not others \cite{BBW}.

There are (at least) two natural ways to generalize our CUE-Schur construction:
first, one may consider disordered Jack measures whose parameters
are sampled from the Circular Beta Ensemble; second, one may consider dynamical disordered Schur measures whose parameters are
the eigenvalues of Brownian motion on the unitary group. Since the primary goal of the 
present paper is to identify a new paradigm --- random partitions with random matrix disorder ---
we focus solely on the simplest such model and leave the broader exploration of this new landscape to 
future work \cite{NiNo}. 

\subsubsection*{Acknowledgement}
I thank Gerard Ben Arous for a friendly lesson in spin glass basics.

\section{Schur Measures}
In this background section we introduce (nonrandom) Schur measures.

\subsection{Partitions and particles}
For each $N \in \N,$ let $\Y_N$ denote the set of nonnegative integer vectors $\lambda=(\lambda_1,\dots,\lambda_N)$
with weakly decreasing coordinates and let $\X_N$ denote the set of nonnegative integer 
vectors $\kappa=(\kappa_1,\dots,\kappa_N)$ with strictly decreasing coordinates. 
A vector $\lambda \in \Y_N$ can be viewed as a partition of the number

    \begin{equation}
        |\lambda|=\lambda_1+\dots+\lambda_N
    \end{equation}

\noindent
having at most $N$ nonzero parts, or equivalently as a Young diagram with exactly $|\lambda|$ cells and at most 
$N$ rows. A vector $\kappa \in \X_N$ can be viewed as an $N$-element set 

    \begin{equation}
        \{\kappa_1,\dots,\kappa_N\} \subset \N_0,
    \end{equation}

\noindent
or equivalently as a configuration of $N$ fermions on $\N_0.$ 
There is a bijection 
$\Y_N \to \X_N$ which sends $\lambda=(\lambda_1,\dots,\lambda_N)$
to $\kappa(\lambda)=(\kappa_1(\lambda),\dots,\kappa_N(\lambda))$ with coordinates

    \begin{equation}
        \kappa_i(\lambda) = N-i + \lambda_i, \quad 1 \leq i \leq N.
    \end{equation}

\noindent
The bijection $\lambda \mapsto \kappa(\lambda)$ parameterizes particle configurations by displacement from
the extremal fully-packed configuration $(N-1,N-2,\dots,1,0).$

Each $\lambda \in \Y_N$ also corresponds to an irreducible polynomial representation
$(\W_N^\lambda,S_N^\lambda)$ of the unitary group, $\group{U}_N$. The corresponding
character,

    \begin{equation}
        s_\lambda(U) = \Tr S_N^\lambda(U),
    \end{equation}

\noindent
is a symmetric polynomial in the eigenvalues $u_1,\dots,u_N$ of
the unitary matrix $U$, homogeneous
of degree $|\lambda|.$ This is the Schur polynomial corresponding to $\lambda \in \Y_N$,
and we have \cite{Bump}

    \begin{equation}
        s_\lambda(U) = \frac{\det[u_j^{N-i+\lambda_i}]}{\det[u_j^{N-i}]}.
    \end{equation}

\subsection{Schur measures} 
Let $(U_N)_{N=1}^\infty$ be a sequence of unitary matrices with $U_N \in \group{U}_N.$
For each $N \in \N,$ the corresponding Schur measure $\P_N=\P_N(q,U_N)$ on $\Y_N$ is defined by

    \begin{equation}
    \label{eqn:SchurMeasure}
        \P_N(\lambda) = \frac{q^{|\lambda|}}{Z_N}|s_\lambda(U_N)|^2, 
    \end{equation}

\noindent
where $q \in (0,1)$ is a fugacity parameter. While $\P_1$ is 
simply the geometric distribution on $\Y_1=\N_0$ with parameter $q$, for $N >1$ the environment
$U_N$ matters. The partition function

    \begin{equation}
        Z_N=\sum_{\lambda \in \Y_N}q^{|\lambda|}|s_\lambda(U_N)|^2 = 
        \sum_{\lambda \in \Y_N}q^{|\lambda|} s_\lambda(U_N)s_\lambda(U_N^{-1})
    \end{equation}

\noindent
is given by Cauchy's generalization of the geometric series \cite{Bump},

    \begin{equation}
    \label{eqn:Cauchy}
        Z_N=\frac{1}{\det(I-qU_N \otimes U_N^{-1})}=\prod_{i,j=1}^N \frac{1}{1-qu_i\overline{u}_j}.
    \end{equation}

If we instead view $\P_N$ as a probability measure on $\X_N,$ we are looking at an
$N$-point process in $\N_0=\{0,1,2,\dots\}.$ This process is determinantal:
rewriting \eqref{eqn:SchurMeasure} as

    \begin{equation}
        \P_N(\kappa) = \frac{1}{\widetilde{Z}_N}\det[u_j^{\kappa_i}]\det[\overline{u}_j^{\kappa_i}]
        \prod_{i=1}^N q^{\kappa_i},
    \end{equation}

\noindent
where the rescaled partition function is

    \begin{equation}
        \widetilde{Z}_N=Z_N \prod_{i<j} q|u_i-u_j|^2,
    \end{equation}

\noindent
makes it clear that we have a biorthogonal particle ensemble in the 
sense of Borodin \cite{Borodin}. 
From Cauchy's identity \eqref{eqn:Cauchy}, the free
energy of the ensemble is

    \begin{equation}
    \label{eqn:Miwa}
        \log Z_N = \sum_{d=1}^\infty \frac{q^d}{d} |\Tr U_N^d|^2.
    \end{equation}

\subsection{Cyclotomic Schur measures}
As an example, let $(U_N)_{N=1}^\infty$ be a sequence of unitary matrices such that the 
eigenvalues of $U_N$ are the distinct $N$th roots of unity. In this case, 

    \begin{equation}
        \Tr U_N^d = \begin{cases}
            N, \text{ if } d \text{ divisible by }N \\
            0, \text{ otherwise}
        \end{cases},
    \end{equation}

\noindent
and hence

    \begin{equation}
        Z_N=\frac{1}{(1-q^N)^N}.
    \end{equation}

\noindent
The determinant $\det[u_j^{\kappa_i}]$ vanishes unless reducing each coordinate of 
$\kappa=(\kappa_1,\dots,\kappa_N)$ modulo $N$ leaves a residue vector
whose coordinates are a permutation of $\{0,1,\dots,N-1\}.$ When this is so, $s_\lambda(U_N)$ equals the 
sign of that permutation. Thus, our cyclotomic Schur measure $\P_N$
on $\Y_N$ is supported on partitions $\lambda$ with empty $N$-core, and in fact it has density
$\P_N(\lambda) = (1-q^N)^Nq^{|\lambda|}$ on this set. The product measure $\otimes_{N=1}^\infty \P_N$ may be 
naturally identified with 
the $q^{\text{volume}}$ Gibbs measure on plane partitions \cite{NiNo}.

\section{Schur Measures with CUE Disorder}
    Let $(U_N)_{N=1}^\infty$ be the Circular Unitary Ensemble \cite{Forrester}: $U_N$
    is a random unitary matrix whose distribution in $\group{U}_N$ is Haar measure, $\mathbb{H}_N$. 
    For each $N \in \N$, the corresponding Schur measure $\P_N=\P(q,U_N)$ 
    is now a random probability measure on $\Y_N$ which we call the CUE-Schur measure on $\Y_N$. 
    Equivalently, if we view $\P_N$ as 
    a random measure on $\X_N$ we are looking at a disordered determinantal $N$-point process in $\N_0$ 
    whose correlation kernel depends on the random environment $U_N$ through its eigenvalues $u_1,\dots,u_N$. 

    \subsection{Expectations}
    The partition function and free energy,

        \begin{equation}
            Z_N = \sum_{\lambda \in \Y_N}q^{|\lambda|}|s_\lambda(U_N)|^2
            \quad\text{and}\quad
            \log Z_N = \sum_{d=1}^\infty \frac{q^d}{d}|\Tr U_N^d|^2,
        \end{equation}

    \noindent
    are now random variables. Let $\E$ denote Haar expectation on $\group{U}_N.$
    
        \begin{thm}
        \label{thm:ExpectationValues}
            We have

                \begin{equation*}
                    \E Z_N = \prod_{n=1}^N \frac{1}{1-q^n}
                \end{equation*}

            \noindent
            and 

                \begin{equation*}
                    \E \log Z_N = q\frac{1-q^N}{1-q}+N\sum_{d=N+1}^\infty \frac{q^d}{d}.
                \end{equation*}
        \end{thm}

    \begin{proof}
    These evaluations follow from a pair of orthogonality formulas,

        \begin{equation*}
            \E s_\lambda(U_N)s_\mu(U_N^{-1})=\delta_{\lambda\mu}  \quad\text{and}\quad 
            \E\Tr U_N^c \Tr U_N^{-d} = \delta_{cd}\min(d,N),
        \end{equation*}

    \noindent
    due to Schur and Diaconis-Evans \cite{DE}, respectively. Schur orthogonality is 
    a classical result which holds for 
    irreducible characters of arbitrary compact groups \cite{Bump}, whereas Diaconis-Evans orthogonality
    is a special feature of unitary group averages tied to Schur-Weyl duality \cite{Bump}. 
    For the expected partition function, Schur orthogonality gives

            \begin{equation*}
                \E Z_N = \sum_{\lambda \in \Y_N} q^{|\lambda|}\E|s_\lambda(U_N)|^2=\sum_{\lambda \in \Y_N} q^{|\lambda|},
            \end{equation*}

        \noindent 
        the generating function for partitions with at most $N$ parts.
        For the expected free energy, Diaconis-Evans orthogonality gives

            \begin{equation*}
                \E \log Z_N = \sum_{d=1}^\infty \frac{q^d}{d}\E|\Tr U_N^d|^2 = 
                \sum_{d=1}^N q^d+N\sum_{d=N+1}^\infty \frac{q^d}{d},
            \end{equation*}

        \noindent
        a finite geometric series plus a logarithmic tail. 
    \end{proof}

    By Jensen's inequality, we have

        \begin{equation}
        \label{eqn:Jensen}
            \E \log Z_N \leq \log \E Z_N
        \end{equation}

    \noindent
    and the inequality is strict for $N>1.$ The basic disordered systems question is whether this
    gap between quenched and annealed free energy persists in the thermodynamic limit. 

    \begin{thm}
        \label{thm:DisorderGap}
        We have

            \begin{equation*}
                \lim_{N \to \infty} \log \E Z_N=\sum_{n=1}^\infty \frac{q^n}{1-q^n}\frac{1}{n}
                \quad\text{and}\quad
                \lim_{N \to \infty} \E\log Z_N = \frac{q}{1-q}.
            \end{equation*}
    \end{thm}

    \begin{proof}
    First note that $Z_N$ and $\log Z_N$ extend to random analytic functions on the unit disc 
    $\D=\{q \in \C \colon |q|<1\}.$ For the expected partition function, Theorem \ref{thm:ExpectationValues} gives

                \begin{equation*}
                    \lim_{N \to \infty} \E Z_N = \lim_{N \to \infty} \prod_{n=1}^N \frac{1}{1-q^n} = \prod_{n=1}^\infty \frac{1}{1-q^n},
                \end{equation*}

    \noindent
    where the convergence is uniform on compact subsets of $\D$. 
    We thus have 

        \begin{equation*}
             \lim_{N \to \infty} \log \E Z_N = - \log \phi(q),
        \end{equation*}

    \noindent
    where $\phi(q)= \prod_{n=1}^\infty (1-q^n)$ is Euler's function \cite{Apostol}
     and 

        \begin{equation}
           - \log \phi(q) = \sum_{n=1}^\infty \frac{q^n}{1-q}\frac{1}{n}.
        \end{equation}

    \noindent
    For the 
    expected free energy, Theorem \ref{thm:ExpectationValues} gives

                \begin{equation*}
                    \sum_{d=1}^N q^d < \E \log Z_N<\sum_{d=1}^N q^d+\frac{q^{N+1}}{1-q},
                \end{equation*}

    \noindent
    so that

        \begin{equation*}
            \lim_{N \to \infty} \E \log Z_N=\frac{q}{1-q}
        \end{equation*}

    \noindent
    uniformly on compact subsets of $\D.$ 
    \end{proof}
    
    We conclude from Theorem \ref{thm:DisorderGap} that the limiting disorder gap 
    for the CUE-Schur ensemble is 

        \begin{equation*}
            \lim_{N \to \infty} (\log \E Z_N - \E \log Z_N)= \sum_{n=2}^\infty \frac{q^n}{1-q}\frac{1}{n}.
        \end{equation*}

    \noindent
    This difference can be explicitly computed for particular
    choices of the fugacity $q \in (0,1)$. Indeed, various special values of $\phi(q)$ were computed explicitly by Ramanujan
    \cite{Berndt}. For example, at fugacity $q=e^{-\pi}$ the limiting disorder gap is 

        \begin{equation}
        \label{eqn:RamanujanValue}
            \frac{7}{8}\log 2+\frac{3}{4}\log \pi-\log\Gamma(\frac{1}{4})-\frac{\pi}{24} -\frac{1}{e^\pi-1}
            \approx 9.63 \times 10^{-4}.
        \end{equation}

    \subsection{Moments of $Z_N$}
    Consider the adjoint representation associated to $(\W_N^\lambda,S_N^\lambda),$ where
    $U \in \group{U}_N$ acts on $A \in \End \W_N^\lambda$ according to $S_N^\lambda(U)AS_N^\lambda(U^{-1}).$
    For any $k \in \N$ and any $\lambda^1,\dots,\lambda^k \in \Y_N$, let $I_N(\lambda^1,\dots,\lambda^k)$
    denote the dimension of the space of $\group{U}_N$-invariants in the corresponding tensor product of 
    adjoint representations,

        \begin{equation}
        \label{eqn:AdjointProduct}
            \End \W_N^{\lambda^1} \otimes \dots \otimes \End \W_N^{\lambda^k}.
        \end{equation}

        \begin{thm}
            \label{thm:PartitionFunctionMoments}
            We have

                \begin{equation*}
                    \E Z_N^k = \sum_{\lambda^1,\dots,\lambda^k \in \Y_N}q^{|\lambda^1|+\dots+|\lambda^k|}
            I_N(\lambda^1,\dots,\lambda^k).
                \end{equation*}
        \end{thm}

        \begin{proof}
            The $k$th power of the partition function is

                \begin{equation*}
                    Z_N^k = \sum_{\lambda^1,\dots,\lambda^k \in \Y_N}q^{|\lambda^1|+\dots+|\lambda^k|}
                    |s_{\lambda^1}(U_N)|^2 \dots |s_{\lambda^k}(U_N)|^2.
                \end{equation*}

            \noindent
            For any $U \in \U_N$, the product 

                \begin{equation*}
                    |s_{\lambda^1}(U)|^2 \dots |s_{\lambda^k}(U)|^2=
                    s_{\lambda^1}(U)s_{\lambda^1}(U^{-1}) \dots s_{\lambda^k}(U)s_{\lambda^k}(U^{-1})
                \end{equation*}

            \noindent
            is the character of $U$ acting in the representation \eqref{eqn:AdjointProduct}.
            Since averaging the character of a representation gives the dimension of the space of 
            invariants in that representation \cite{Bump}, the result follows.
        \end{proof}

    \subsection{Limit of $\log Z_N$}
    We will now determine the distribution of the free energy
    in the thermodynamic limit.

        \begin{thm}
        \label{thm:LimitDensity}
        For any $q \in (0,1)$, we have 

            \begin{equation*}
                \log Z_N \implies \sum_{d=1}^\infty q^d X_d
            \end{equation*}

        \noindent
        as $N \to \infty$, 
        where $X_1,X_2,X_3,\dots$ are iid $\mathrm{Exp}(1)$
        random variables.
        \end{thm}

        \begin{proof}
        We use Johansson's multivariate central limit theorem for CUE traces \cite{Johansson}: 
        for any fixed $M \in \N$ we have 
        
                \begin{equation*}
                    (\Tr U_N,\Tr U_N^2,\dots,\Tr U_N^M) \implies
                    (\sqrt{1}Z_1,\sqrt{2}Z_2,\dots,\sqrt{M}Z_M)
                \end{equation*}

        \noindent
        as $N \to \infty,$ where the convergence is in distribution and
        $Z_1,\dots,Z_M$ are iid standard complex Gaussians. Since the law of
        $|Z_1|^2$ is $\mathrm{Exp}(1)$, Johansson's theorem
        implies that for any fixed $M \in \N$ the truncated free energy

                \begin{equation*}
                    L_{MN}= \sum_{d=1}^M \frac{q^d}{d}|\Tr U_N^d|^2 
                \end{equation*}

            \noindent
            satisfies

                \begin{equation*}
                    L_{MN} \implies \sum_{d=1}^M q^d X_d
                \end{equation*}

            \noindent
            as $N \to \infty$, where $X_1,\dots,X_M$ are iid $\mathrm{Exp}(1)$.
            To conclude $N \to \infty$ convergence in distribution of the full free energy $\log Z_N$ 
            to $F=\sum_{d=1}^\infty q^d X_d$ from this truncation, we must show that no mass
            escapes into the tail sum

                \begin{equation*}
                    T_{MN} = \sum_{d=M+1}^\infty \frac{q^d}{d}|\Tr U_N^d|^2
                \end{equation*}
                    
        \noindent
        as $N \to \infty$. By Diaconis-Evans we have

            \begin{equation*}
              \E T_{MN}  = \sum_{d=M+1}^\infty \frac{q^d}{d}\min(d,N) < \sum_{d=M+1}^\infty q^d, 
            \end{equation*}

        \noindent
       so Markov's inequality gives

            \begin{equation*}
                \sup_{N \in \N}\mathbb{H}_N(T_{MN}>\varepsilon) \leq \frac{1}{\varepsilon}\frac{q^{M+1}}{1-q}
            \end{equation*}

        \noindent
        for any $\varepsilon > 0,$ where once again $\mathbb{H}_N$ denotes the Haar probability measure on $\U_N.$
        \end{proof}

    Theorem \ref{thm:LimitDensity} can be sharpened to a finite $N$ bound on the distance
    between the distributions of $\log Z_N$ and $F=\sum_{d=1}^\infty q^dX_d$ using results of Johansson and 
    Lambert \cite{JohLam}, but this is not necessary for our purposes.

    \section{Extensive Scaling}
    \label{sec:Extensive}
    We now consider a scaling regime for CUE-Schur measure 
    in which the free energy becomes extensive and the free energy density self-averages.
    Throughout this section, $U_N$ denotes a random unitary matrix whose distribution in 
    the unitary group $\group{U}_N$ is Haar measure $\H_N$, and $c>0$ is an arbitrary but 
    fixed positive constant. 

    \subsection{Near-critical scaling}
    The free energy $\log Z_N = \sum_{d=1}^\infty \frac{q^d}{d}|\Tr U_N^d|^2$ converges 
    for any fugacity $q \in (0,1)$, but diverges at $q=1.$ We shall consider the large $N$
    asymptotics of $\log Z_N$ when the fugacity $q=q_N$ depends critically on the particle number 
    $N$, meaning that $q_N \to 1$ as $N \to \infty$. We implement this scaling as $q_N=e^{-\frac{c}{N}}.$ 
    In this critical scaling, 
    the quenched and annealed free energies of the CUE-Schur measure $\P_N=\P_N(e^{-\frac{c}{N}},U_N)$ 
    are asymptotically proportional to the particle number.

        \begin{thm}
        \label{thm:ExtensiveExpectation}
            The near-critical quenched and annealed free energies satisfy 

                \begin{equation*}
                    \E \log Z_N \sim \mu_c N
                    \quad\text{and}\quad
                    \log \E Z_N \sim \nu_c N
                \end{equation*}

            \noindent
            as $N \to \infty,$ where 

                \begin{equation*}
                \mu_c= \frac{1}{c}\int_0^c e^{-x}\mathrm{dx} + \int_c^\infty \frac{e^{-x}}{x}\mathrm{d}x
                \quad\text{and}\quad
                    \nu_c = \frac{1}{c}\int_0^c \log \frac{1}{1-e^{-x}}\mathrm{d}x
                \end{equation*}

        \end{thm}

        \begin{proof}
            Since we have access to the exact formulas of Theorem \ref{thm:ExpectationValues} 
            this is computational. 
            
            For any fixed $q \in (0,1)$ and finite $N \in \N$, 
            the quenched free energy per particle is

                \begin{equation*}
                    \frac{1}{N}\E \log Z_N= \frac{q}{N}\frac{1-q^N}{1-q}+\sum_{d=N+1}^\infty 
                    \frac{q^d}{d}.
                \end{equation*}

            \noindent
            Replacing $q$ with $q_N=e^{-\frac{c}{N}}$ in the finite geometric series component gives

                \begin{equation*}
                    \frac{e^{-\frac{c}{N}}}{N}\frac{1-e^{-c}}{1-e^{-\frac{c}{N}}}=
                    \frac{1-e^{-c}}{\frac{c}{1!}+\frac{c^2}{2!N}+\frac{c^3}{3!N^2}+\dots},
                \end{equation*}

            \noindent
            which converges to 

                \begin{equation*}
                    \frac{1-e^{-c}}{c}=\frac{1}{c}\int_0^c e^{-x}\mathrm{dx}
                \end{equation*}

            \noindent
            as $N \to \infty.$ Replacing $q$ with $q_N=e^{-\frac{c}{N}}$ in the logarithmic tail component
            gives

                \begin{equation}
                    \sum_{d=N+1}^\infty \frac{e^{-\frac{cd}{N}}}{d} = \sum_{d=N+1}^\infty\frac{1}{N}
                    h\left(\frac{d}{N}\right), \quad h(x)=\frac{e^{-cx}}{x},               
                \end{equation}

            \noindent
            which intuitively is a Riemann sum for the improper integral 

                \begin{equation}
                    \int_1^\infty \frac{e^{-cx}}{x}\mathrm{d}x = \int_c^\infty \frac{e^{-x}}{x}\mathrm{d}x.
                \end{equation}

            \noindent
            Using the fact that $h(x)$ is a continuous, positive, decreasing function on $[1,\infty),$
            it is straightforward calculus to check that indeed

                \begin{equation*}
                    \lim_{N \to \infty} \sum_{d=N+1}^\infty \frac{e^{-\frac{cd}{N}}}{d} = \int_1^\infty \frac{e^{-cx}}{x}\mathrm{d}x.
                \end{equation*}
            
        The argument is similar for the annealed free energy density: replacing $q$ with 
        $q_N=e^{-\frac{c}{N}}$ in the exact formula
        
            \begin{equation*}
                \frac{1}{N} \log \E Z_N = -\frac{1}{N} \sum_{n=1}^N \log(1-q^n).
            \end{equation*}

        \noindent
        yields 

            \begin{equation*}
                \sum_{n=1}^N \frac{1}{N}\log\frac{1}{1-e^{-\frac{cn}{N}}} = \sum_{n=1}^N 
                \frac{1}{N}g\left(\frac{n}{N}\right), \quad g(x)=-\log(1-e^{-cx}).
            \end{equation*}

        \noindent
        This is a Riemann sum for the integral 

            \begin{equation}
                -\int_0^1 \log(1-e^{-cx}) \mathrm{d}x = -\frac{1}{c}\int_0^c \log(1-e^{-x}) \mathrm{d}x.
            \end{equation}

        \noindent
        Using the fact that $g(x)$ is a continuous, positive, decreasing function on $(0,1],$
        and taking suitable care with the logarithmic singularity at $x=0,$ it is again 
        calculus to verify

            \begin{equation*}
                 \lim_{N \to \infty}\sum_{n=1}^N \frac{1}{N}\log\frac{1}{1-e^{-\frac{cn}{N}}} = 
                 \int_0^1 \log\frac{1}{1-e^{-cx}}\mathrm{d}x.
            \end{equation*}
        \end{proof}
        
    \subsection{Variance of $\log Z_N$}
    Our goal now is to show that the free energy density  concentrates around its mean in our near-critical
    scaling, so that we have self-averaging for near-critical CUE-Schur measures in the thermodynamic limit. 
    We first take a step back and show that $\log Z_N$ is a pair statistic of the 
    CUE.

        \begin{thm}
            \label{thm:PairStatistic}
            At fixed fugacity $q \in (0,1)$ we have

                \begin{equation*}
                    \log Z_N = \frac{1}{2} \sum_{m,n=1}^N f(\theta_m-\theta_n),
                \end{equation*}

            \noindent
            where $e^{i\theta_1},\dots,e^{i\theta_N}$ are the eigenvalues of $U_N$ and 

                \begin{equation*}
                    f(\theta) = - \log(1-2q\cos\theta+q^2)
                \end{equation*}

            \noindent
            is an even periodic function with Fourier series

                \begin{equation*}
                    f(\theta) = 2\sum_{k=1}^\infty \frac{q^k}{k}\cos(k\theta).
                \end{equation*}
        \end{thm}

        \begin{proof}
            From \eqref{eqn:Miwa}, we have 

            \begin{equation*}
                \log Z_N = \sum_{d=1}^\infty \frac{q^d}{d}\sum_{m,n=1}^N e^{id(\theta_m-\theta_n)}=
                -\sum_{m,n=1}^N \log (1-qe^{i(\theta_m-\theta_n)}).
            \end{equation*}

        \noindent 
        Each of the $N$ diagonal terms in the double sum contributes the constant $-\log (1-q),$
        and pairing off-diagonal terms gives

            \begin{equation*}
                -\log(1-qe^{i(\theta_m-\theta_n)}) - \log(1-qe^{-i(\theta_m-\theta_n)})
                =-\log (1-2q\cos(\theta_m-\theta_n) +q^2).
            \end{equation*}

        \noindent
        Noting that the function

            \begin{equation*}
                f(\theta)= -\log (1-2q\cos \theta+q^2)
            \end{equation*}

        \noindent
        satisfies $f(0) = -2 \log(1-q)$, we thus have

             \begin{equation*}
                \log Z_N =\frac{1}{2}\sum_{m,n=1}^N f(\theta_m-\theta_n).
            \end{equation*}

        \noindent
        Concerning the Fourier series of $f(\theta),$ we have

            \begin{equation*}
                f(\theta)= \log \frac{1}{1-qe^{-i\theta}} + \log \frac{1}{1-qe^{i\theta}}
                = \sum_{k=1}^\infty \frac{q^{k}}{k}(e^{ik\theta}+e^{-ik\theta})
                =2\sum_{k=1}^\infty \frac{q^k}{k} \cos(k\theta).
            \end{equation*}

    \end{proof}

    Although Theorem \ref{thm:PairStatistic} is an elementary rewriting of the trace expansion 
    of $\log Z_N$ coming from the Cauchy identity, the fact that the free energy of CUE-Schur measure
    is a pair statistic of the CUE is actually quite important. It allows us to connect with 
    work of Soshnikov and Wu \cite{SoWu}, whose studied general pair statistics 

        \begin{equation}
            S_N(f) = \sum_{m,n=1}^N f(\theta_m-\theta_n)
        \end{equation}

    \noindent
    of CUE eigenvalues and obtained the following crucial variance formula. 

    \begin{thm}
        \label{thm:PairVarianceFormula}
        For real even test functions

            \begin{equation*}
                f(\theta) = \widehat{f}(0)+2\sum_{k=1}^\infty \widehat{f}(k)\cos(k\theta),
            \end{equation*}

        \noindent
        we have

             \begin{equation*}
                        \frac{1}{4}\mathrm{Var}[S_N(f)]= A_N(f) + B_N(f) -C_N(f)-D_N(f),
                    \end{equation*}

                \noindent
                where 

                    \begin{equation*}
                        \begin{split}
                            A_N(f) &= \sum_{k=1}^{N-1} k^2 \widehat{f}(k)^2\\
                            B_N(f) & = (N^2-N)\sum_{k=N}^\infty \widehat{f}(k)^2 \\
                            C_N(f) & = \sum_{\substack{\max(k,l) \geq N \\ 1 \leq |k-l| \leq N-1}} (N-|k-l|)\widehat{f}(k)\widehat{f}(l)\\
                            D_N(f) &= \sum_{\substack{1 \leq k,l \leq N-1 \\ k+l \geq N+1}} (k+l-N)\widehat{f}(k)\widehat{f}(l).
                        \end{split} 
                    \end{equation*}
    \end{thm}

    Combining Theorems \ref{thm:PairStatistic} and \ref{thm:PairVarianceFormula}, we can evaluate
    the variance of the free energy for CUE-Schur measures as follows 
        
         \begin{cor}
                \label{cor:VarianceFormula}
                For any fixed $q \in (0,1)$ and finite $N \in \N,$ we have

                    \begin{equation*}
                        \mathrm{Var}[\log Z_N] = A_N+B_N-C_N-D_N,
                    \end{equation*}

                \noindent
                where

                    \begin{equation*}
                        \begin{split}
                            A_N &= \sum_{k=1}^{N-1} q^{2k}\\
                            B_N & = N(N-1)\sum_{k=N}^\infty \frac{q^{2k}}{k^2} \\
                            C_N & = \sum_{\substack{\max(k,l) \geq N \\ 1 \leq |k-l| \leq N-1}} (N-|k-l|)\frac{q^{k+l}}{kl} \\
                            D_N &= \sum_{\substack{1 \leq k,l \leq N-1 \\ k+l \geq N+1}} (k+l-N)\frac{q^{k+l}}{kl}.
                        \end{split} 
                    \end{equation*}
            \end{cor}

        \subsection{Self-averaging}
        We are now in position to prove $N \to \infty$ self-averaging of the free energy density 
        $\frac{1}{N}\log Z_N$ in the near-critical scaling where $q_N=e^{-\frac{c}{N}}.$

            \begin{thm}
                \label{thm:SelfAveraging}
                For any $\varepsilon >0$ we have

                    \begin{equation*}
                        \lim_{N \to \infty} \mathbb{H}_N \left( |N^{-1}\log Z_N - \mu_c| > \varepsilon \right)=0,
                    \end{equation*}

                \noindent
                where 

                    \begin{equation*}
                        \mu_c = \frac{1-e^{-c}}{c}+\int_1^\infty \frac{e^{-cx}}{x}\mathrm{d}x.
                    \end{equation*}

            \end{thm}

            \begin{proof}
                Since we have already proved in Theorem \ref{thm:ExtensiveExpectation} that 
                this limit holds in expectation, it remains only to show that

                    \begin{equation*}
                        \lim_{N \to \infty }\frac{1}{N^2}\mathrm{Var}[\log Z_N]=0,
                    \end{equation*}

                \noindent
                and this follows from Corollary \ref{cor:VarianceFormula}. 
                
                First, for any $q \in (0,1)$ and $N \in \N$ we have

                    \begin{equation*}
                        \frac{1}{N}A_N = \frac{1}{N}\sum_{k=1}^{N-1}q^{2k}< q^2 <1,
                    \end{equation*}

                \noindent
                so 

                    \begin{equation*}
                        \lim_{N \to \infty}\frac{1}{N^2}A_N=0
                    \end{equation*}

                \noindent
                no matter how $q$ scales with $N$. 

                Second, for any $q \in (0,1)$ and $N \in \N$ we have

                    \begin{equation*}
                        \frac{1}{N}B_N = (N-1)\sum_{k=N}^\infty \frac{q^{2k}}{k^2}<
                        (N-1)\sum_{k=N}^\infty \frac{1}{k^2} < (N-1) \int_{N-1}^\infty\frac{\mathrm{d}x}{x^2}=1,
                    \end{equation*}

                \noindent
                so 

                    \begin{equation*}
                        \lim_{N \to \infty}\frac{1}{N^2}B_N=0
                    \end{equation*}

                \noindent
                no matter how $q$ scales with $N$. 

                Third, for any $q \in (0,1)$ and $N \in \N$ each term of the double sum $C_N$
                is symmetric in its indices and so the sum can be written

                    \begin{equation*}
                        C_N=2\sum_{d=1}^{N-1}(N-d)\sum_{l=N-d}^\infty \frac{q^{2l+d}}{l(l+d)}.
                    \end{equation*}

                \noindent
                For any $d \in \{1,\dots,N-1\}$ and $l \geq N-d$, we have $l(l+d) \geq N(N-d)$ and hence
                $\frac{N-d}{l(l+d)} \leq \frac{1}{N}$, which gives the bound

                    \begin{equation*}
                        C_N \leq \frac{2}{N}\sum_{d=1}^{N-1}\sum_{l=N-d}^\infty q^{2l+d}.
                    \end{equation*}

                \noindent
                Summing the inner geometric series,

                    \begin{equation*}
                        \sum_{l=N-d}^\infty q^{2l+d} = \frac{q^{2N-d}}{1-q^2},
                    \end{equation*}

                \noindent
                we have

                    \begin{equation*}
                        C_N \leq \frac{2}{N(1-q^2)}\sum_{d=1}^{N-1}q^{2N-d}.
                    \end{equation*}

                \noindent
                The remaining finite geometric sum is

                    \begin{equation*}
                        \sum_{d=1}^{N-1}q^{2N-d} = q^{2N-1} + \dots + q^{N+1}< q^{N+1}N
                    \end{equation*}

                \noindent
                and we arrive at the bound

                    \begin{equation*}
                        C_N \leq \frac{2q^{N+1}}{(1-q^2)}.
                    \end{equation*}

                \noindent
                Replacing $q$ with $q_N=e^{-\frac{c}{N}}$ we get

                    \begin{equation*}
                        C_N \leq \frac{2e^{-c}}{e^{\frac{c}{N}}-e^{-\frac{c}{N}}} = \frac{e^{-c}}{\sinh(\frac{c}{N})} \sim \frac{e^{-c}}{c}N,
                    \end{equation*}

                \noindent
                so 

                    \begin{equation*}
                        \lim_{N \to \infty} \frac{1}{N^2}C_N=0.
                    \end{equation*}

                Fourth, for any $q \in (0,1)$ and $N \in \N$ the double sum $D_N$ has indices in the range $1 \leq k,l \leq N-1$, so 
                $(N-k)(N-l)>0$. On the other hand, $(N-k)(N-l) = kl-N(k+l-N),$ so

                    \begin{equation*}
                        \frac{k+l-N}{kl}< \frac{1}{N}
                    \end{equation*}

                \noindent
                and we get the bound

                    \begin{equation*}
                        D_N < \frac{1}{N}\sum_{\substack{1 \leq k,l \leq N-1 \\ k+l \geq N+1}}q^{k+l}.
                    \end{equation*}

                \noindent
                Now group the terms of this double sum according to the value $m=k+l$, which is a 
                number in $\{N+1,\dots,2N-2\}.$ The number of pairs $(k,l)$ with $1 \leq k,l \leq N-1$
                and $k+l=m$ is $2N-1-m$, so

                    \begin{equation*}
                        D_N < \frac{1}{N}\sum_{m=N+1}^{2N-2}(2N-1-m)q^m\leq \frac{N-2}{N}\sum_{m=N+1}^{2N-2}q^m.
                    \end{equation*}

                \noindent
                We therefore have

                    \begin{equation*}
                        D_N < \sum_{m=N+1}^{2N-2}q^m < (N-2)q^{N+1}.
                    \end{equation*}

                \noindent
                Replacing $q$ with $q_N=e^{-\frac{c}{N}}$ we have

                    \begin{equation}
                        D_N < (N-2)e^{-c}e^{-\frac{c}{N}} \sim e^{-c}N,
                    \end{equation}
                
                \noindent
                so

                    \begin{equation*}
                        \lim_{N \to \infty}\frac{1}{N^2}D_N=0.
                    \end{equation*}

            \end{proof}

            The proof of Theorem \ref{thm:SelfAveraging} used only elementary manipulations
            of explicit sums to show that each of the constituents $A_N,B_N,C_N,D_N$ of $\mathrm{Var}[\log Z_N]$
            is $O(N).$ Comparing each of the quantities $\frac{1}{N}A_N(c),\frac{1}{N}B_N(c),\frac{1}{N}C_N(c),\frac{1}{N}D_N(c)$ 
            with Riemann sums as in 
            Theorem \ref{thm:ExtensiveExpectation} yields a more precise result. 

            \begin{thm}
                In the near-critical scaling $q_N=e^{-\frac{c}{N}},$ we have

                    \begin{equation*}
                        \mathrm{Var}[\log Z_N]\ \sim (\alpha_c+\beta_c-\gamma_c-\delta_c)N
                    \end{equation*}

                \noindent
                as $N \to \infty,$ where 

                    \begin{equation*}
                    \begin{split}
                            \alpha_c &= \int_0^1 e^{-2cx}\mathrm{d}x=\frac{1-e^{-2c}}{2c} \\
                            \beta_c &= \int_1^\infty \frac{e^{-2cx}}{x^2}\mathrm{d}x \\
                            \gamma_c &= \iint\limits_{\substack{x,y \geq 0 \\ \max(x,y) \geq 1\\ |x-y|\leq 1}} \frac{1-|x-y|}{xy} e^{-c(x+y)}\mathrm{d}x\mathrm{d}y \\
                            \delta_c &= \iint\limits_{\substack{0 \leq x,y \leq 1 \\ x+y \geq 1}} \frac{x+y-1}{xy} e^{-c(x+y)}\mathrm{d}x\mathrm{d}y.
                    \end{split}
                    \end{equation*}
            \end{thm}
             
            \begin{proof}
               Normalizing the four constituents of $\mathrm{Var}[\log Z_N]$ by $N$ and 
               writing them in Riemann sum form, we have

                    \begin{equation*}
                        \begin{split}
                            \frac{A_N}{N} &= \frac{1}{N}\sum_{k=1}^{N-1} q_N^{2k}\\
                            \frac{B_N}{N} &= (1-\frac{1}{N})\frac{1}{N}\sum_{k=N}^\infty \frac{q_N^{2k}}{(k/N)^2} \\
                            \frac{C_N}{N} &= \frac{1}{N^2}\sum_{\substack{k,l \geq 1 \\
                                1 \leq |k-l| \leq N-1 \\ \max(k,l) \geq N}} \frac{1-|\frac{k}{N}-\frac{l}{N}|}{(k/N)(l/N)}q_N^{k+l}\\
                            \frac{D_N}{N} &= \frac{1}{N^2}\sum_{\substack{1 \leq k, l \leq N-1 \\k+l \geq N+1}}
                            \frac{\frac{k}{N}+\frac{l}{N}-1}{(k/N)(l/N)}q_N^{k+l}.
                        \end{split} 
                    \end{equation*}

                \noindent
                On the macroscopic scale $k=xN$ and $l=yN$, for $q_N=e^{-\frac{c}{N}}$ we have

                    \begin{equation*}
                        q_N^{2k} = e^{-2cx}
                        \quad\text{and}\quad
                        q_N^{k+l}= e^{-c(x+y)},
                    \end{equation*}

                \noindent
                and the summation ranges converge to the stated regions in $\R$ or $\R^2$.
                Thus the proof is a calculus exercise along the same lines as the proof of Theorem \ref{thm:ExtensiveExpectation}.
            \end{proof}

        \subsection{Free energy fluctuations}
        In the near-critical scaling where fugacity scales with particle number as 
        $q_N=e^{-\frac{c}{N}}$, we have established that the free energy $\log Z_N$ becomes
        extensive, and that its first two cumulants have large $N$ asymptotics

            \begin{equation}
                \E \log Z_N \sim \mu_c N \quad\text{and}\quad \mathrm{Var}[\log Z_N] \sim \sigma_c^2N,
            \end{equation}

        \noindent
        where $\sigma_c = \sqrt{\alpha_c+\beta_c-\gamma_c-\delta_c}$. In fact,
        the following stronger statement holds. 

            \begin{thm}
            \label{thm:CLT}
                As $N \to \infty$, we have

                    \begin{equation*}
                        \frac{\log Z_N - \E \log Z_N}{\sqrt{N}} \implies \mathcal{N}(0,\sigma_c^2).
                    \end{equation*}
            \end{thm}

        Theorem \ref{thm:CLT} can be deduced from the aforementioned work of Soshnikov and Wu \cite{SoWu},
        which analyzes cumulants of pair statistics

            \begin{equation}
                S_N(f) = \sum_{m,n=1}^N f(\theta_m-\theta_n)
            \end{equation}

        \noindent
        of the CUE corresponding to real even test functions 

            \begin{equation}
                f(\theta) = \widehat{f}(0) + 2\sum_{k=1}^\infty \widehat{f}(k) \cos(k\theta)
            \end{equation}

        \noindent
        whose Fourier coefficients satisfy 

            \begin{equation}
            \label{eqn:StaticLimit}
                \lim_{k \to \infty} k \widehat{f}(k) = C \neq 0.
            \end{equation}

        \noindent
        Their asymptotic analysis of the cumulants of $S_N(f)$ shows that 

            \begin{equation}
                \frac{S_N(f)-\E S_N(f)}{\sqrt{N}} \implies \mathcal{N}(0,\sigma_f)
            \end{equation}

        \noindent
        as $N \to \infty.$ The results of \cite{SoWu} apply in our setting, albeit not automatically since we 
        are dealing with the pair statistic $S_N(f_N)$ corresponding to a function

            \begin{equation}
                f_N(\theta) = 2\sum_{k=1}^\infty \frac{q_N^k}{k}\cos(k\theta)=
                2\sum_{k=1}^\infty \frac{e^{-\frac{ck}{N}}}{k}\cos(k\theta)
            \end{equation}

        \noindent
        which varies with $N$. In our situation, the static limit \eqref{eqn:StaticLimit}
        is replaced by the macroscopic profile

            \begin{equation}
            \label{eqn:MacroscopicLimit}
               q_N^{Nx} = e^{-cx}.
            \end{equation}

        \noindent
        However, the cumulant analysis carried out by Soshnikov and Wu \cite{SoWu}
        can be adapted to this scaling, and in particular the suboptimal partitions
        in their combinatorial decomposition remain suppressed simply because $q_N<1$. The only place where the dependence of 
        $f_N = -\log(1-2q_N\cos\theta+q_N^2)$ on $N$ matters is in evaluating the contribution from 
        the asymptotically non-negligible ``optimal'' partitions identified in \cite{SoWu}. In the fixed $f$ setting
        these are controlled by the static limit \eqref{eqn:StaticLimit}, whereas 
        in our $N$-dependent setting they are controlled by the macroscopic profile
        \eqref{eqn:MacroscopicLimit}, but this has no effect on the combinatorial structure of the Soshnikov-Wu argument. 
        In the next section we make use of the law of large numbers for the free energy density $N^{-1}\log Z_N$
        given by Theorem \ref{thm:SelfAveraging}, but we do not need the corresponding central limit theorem 
        given by Theorem \ref{thm:CLT} so we do not give the proof.

    \section{Quenched Limit Laws}
    \label{sec:Quenched}
    We now generalize the results of the previous section in order to obtain a quenched law of large
    numbers for the size of a random sample from near-critical CUE-Schur measure. 

    \subsection{Convergence of derivatives}
    For $U_N$ a Haar unitary, the nonnegative series  

        \begin{equation}
            Z_N(t) = \sum_{\lambda \in \Y_N} e^{-\frac{t|\lambda|}{N}}|s_\lambda(U_N)|^2
        \end{equation}

    \noindent
    and 

         \begin{equation}
            F_N(t) = \frac{1}{N}\sum_{d=1}^\infty 
            \frac{e^{-\frac{td}{N}}}{d}|\Tr U_N^d|^2
        \end{equation}

    \noindent
    converge to define random positive smooth functions of $t \in (0,\infty)$ whose 
    derivatives can be computed by termwise differentiation in $t$.
    Indeed, $Z_N(t)$ is the partition function of the CUE-Schur measure 
    $\P_N=\P_N(e^{-\frac{t}{N}},U_N)$, and $F_N(t)=N^{-1}\log Z_N(t)$ is the 
    corresponding free energy density. The main result of the previous section, Theorem \ref{thm:SelfAveraging},
    gives pointwise convergence of $F_N(t)$ in Haar measure to the limit

        \begin{equation}
            \Phi(t) = \frac{1-e^{-t}}{t}+\int_1^\infty \frac{e^{-tx}}{x}\mathrm{d}x.
        \end{equation}

    The derivatives of the free energy density are the functions

        \begin{equation}
            F_N^{(k)}(t) = \frac{(-1)^k}{N^{k+1}}\sum_{d=1}^\infty d^{k-1}
            e^{-\frac{td}{N}}|\Tr U_N^d|^2, \quad k \in \N.
        \end{equation}

    \noindent
    In particular, $F_N(t)$ is not just convex, it is completely monotone : 
    the function $(-1)^kF_N^{(k)}(t)$ is positive on $(0,\infty)$ for every 
    $k \in \N.$  
    Thus, not only can we use convexity to deduce convergence 
    in probability of $F_N'(t)$ to $\Phi'(t)$ from Theorem \ref{thm:SelfAveraging}, 
    as one does in spin glass theory \cite{Chen},
    we can iterate to do the same for every derivative. 
    
    \begin{prop}
        \label{prop:Derivatives}
        For each fixed $k \in \N$ and every $\varepsilon > 0,$ we have

            \begin{equation*}
                \lim_{N \to \infty} \H_N\left(|F_N^{(k)}(t)-\Phi^{(k)}(t)|>\varepsilon\right)=0
            \end{equation*}

        \noindent
        for all $t \in (0,\infty).$ 
    \end{prop}

    \subsection{Convergence of cumulants}
    As in Section \ref{sec:Extensive}, fix a positive constant $c > 0$ and let $\P_N = \P_N(e^{-\frac{c}{N}},U_N)$
    be the corresponding near-critical CUE-Schur measure on $\Y_N$. Since $\P_N$ is a random measure depending on 
    the disorder $U_N$, ensemble averages

        \begin{equation}
            \langle f\rangle= \sum_{\lambda \in \Y_N} f(\lambda)\P_N(\lambda)
        \end{equation}

    \noindent
    are random variables on $(\group{U}_N,\H_N)$. The random smooth positive function on 
    $(-c,\infty)$ defined by 

        \begin{equation}
            L_N(t) = \frac{Z_N(t+c)}{Z_N(c)}
        \end{equation}

    \noindent
    is just such an ensemble average: we have

        \begin{equation}
            L_N(t) = \sum_{\lambda \in \Y_N} e^{-\frac{t|\lambda|}{N}}\P_N(\lambda) = \langle e^{-tNX_N}\rangle,
        \end{equation}

    \noindent
    where $X_N(\lambda)=N^{-2}|\lambda|.$ The significance of $L_N(t)$ is that
    the $k$th derivative of its logarithm at $t=0$ is $(-1)^kN^kc_k(X_N),$
    where $c_k(X_N)$ is the $k$th cumulant of $X_N$. On the other hand, 

        \begin{equation}
            \log L_N(t) = \log Z_N(t+c)-\log Z_N(c) = NF_N(t+c)-NF_N(c),
        \end{equation}

    \noindent
    so we obtain 

        \begin{equation}
            (-1)^kN^{k-1}c_k(X_N) = F_N^{(k)}(c),
        \end{equation}

    \noindent
    and Proposition \ref{prop:Derivatives} gives a quenched law of large numbers for 
    each fixed cumulant of $X_N.$

        \begin{prop}
            \label{prop:Cumulants}
            For each fixed $k \in \N$ and every $\varepsilon >0,$ we have

            \begin{equation*}
               \lim_{N \to \infty}\H_N\left( \left| N^{k-1}c_k(X_N)- (-1)^k\Phi^{(k)}(c)\right|> \varepsilon\right)=0,
            \end{equation*}

        \noindent
        where 

            \begin{equation*}
                (-1)^k\Phi^{(k)}(c)=\int_0^1 x^ke^{-cx}\mathrm{d}x + \int_1^\infty x^{k-1}e^{-cx}\mathrm{d}x.
            \end{equation*}
        \end{prop}

    The polynomial decay in Haar probability of each cumulant of $X_N(\lambda)=N^{-2}|\lambda|$ beyond the first implies
    a quenched law of large numbers for $X_N$ under near-crtical CUE-Schur measure $\P_N=\P_N(e^{-\frac{c}{N}},U_N).$
    Proposition \ref{prop:Cumulants} gives

        \begin{equation}
           c_1(X_N)= \langle X_N\rangle \longrightarrow -\Phi'(c)=\frac{1-e^{-c}}{c^2}
        \end{equation}

    \noindent
    in Haar measure $\H_N$ as $N \to \infty.$ Furthermore, according to Proposition \ref{prop:Cumulants}
    the variance $c_2(X_N)$ of $X_N$ converges to zero in $\H_N$ as $N \to \infty,$ which gives the following.

        \begin{thm}
            \label{thm:QuenchedLLN}
            For any $\delta>0,$ the random variable 

                \begin{equation*}
                    P_N(\delta) = \P_N\left(\left|X_N-\frac{1-e^{-c}}{c^2} \right|>\delta \right)
                \end{equation*}

            \noindent
            satisfies

                \begin{equation*}
                    \lim_{N \to \infty} \H_N \left( P_N(\delta)>\varepsilon\right)=0
                \end{equation*}

            \noindent
            for every $\varepsilon > 0.$
        \end{thm}

    \section*{AI Tool Disclosure}
    ChatGPT 5.5 Pro was used for proofreading, which led to a significant correction in the stated value of the constant
    $\nu_c$ in Theorem \ref{thm:ExtensiveExpectation}. ChatGPT 5.5 Pro was also used to compute the 
    numerical value stated in Equation \eqref{eqn:RamanujanValue} and to typeset the bibliography. These uses 
    of AI tools notwithstanding, the ideas and results presented in the paper --- as well as 
    the text itself --- were human generated (by the author).

\end{document}